\documentclass[conference]{IEEEtran}
\IEEEoverridecommandlockouts
\usepackage{amsfonts}
\usepackage{amsmath,amssymb}
\usepackage{graphicx,subfigure}
\usepackage{color}      % Need the color package
\usepackage{epsfig}
\usepackage{amssymb}
\usepackage{tipa}
\usepackage{bbm}
\def\ep{\epsilon_n}
\newcommand{\hG} {\hat{G}}
\newcommand{\hV} {\hat{V}}

\long\def\comment#1{} %\input{zcom}

\newtheorem{theorem}{Theorem}[section]
\newtheorem{lemma}[theorem]{Lemma}

\newtheorem{corollary}[theorem]{Corollary}
\newtheorem{definition}[theorem]{Definition}
\begin{document}

\title{A note on the multiple unicast capacity of directed acyclic networks}

\author{\IEEEauthorblockN{Shurui Huang, {\it{Student Member, IEEE}} and Aditya Ramamoorthy, {\it{Member, IEEE}}\\
}
\thanks{This research was supported in part by NSF grant CCF-1018148.}
\authorblockA{Department of Electrical and Computer Engineering\\
Iowa State University, Ames, Iowa 50011\\
Email: \{hshurui, adityar\}@iastate.edu}}

\maketitle
\begin{abstract}
We consider the multiple unicast problem under network coding over
directed acyclic networks with unit capacity edges. There is a set
of $n$ source-terminal ($s_i - t_i$) pairs that wish to
communicate at unit rate over this network. The connectivity
between the $s_i - t_i$ pairs is quantified by means of a
connectivity level vector, $[k_1~k_2 \dots~ k_n]$ such that there
exist $k_i$ edge-disjoint paths between $s_i$ and $t_i$. Our main
aim is to characterize the feasibility of achieving this for
different values of $n$ and $[k_1~ \dots~ k_n]$. For $3$ unicast
connections ($n = 3$), we characterize several achievable and
unachievable values of the connectivity 3-tuple. In
addition, in this work, we have found certain network topologies,
and capacity characterizations that are useful in understanding
the case of general $n$.
\end{abstract}
\section{Introduction}
Network coding has emerged as an interesting alternative to
routing in the next generation of networks. In particular, it is
well-known that the network coding is a provably capacity
achieving strategy for network multicast. The work of \cite{rm}
provides a nice algebraic framework for reasoning about network
coding, and significantly simplifies the proofs of \cite{al}, and
suggests network code design schemes. However,
general network connections, such as multiple unicasts are more difficult
to understand under network coding. In a multiple unicast
connection, there are several source terminal pairs; each source
wishes to communicate to its corresponding terminal. The goal is to
find a characterization of the network resources required to
support this connection using network coding.

The multiple unicast problem has been examined for both directed
acyclic networks \cite{Traskov06}\cite{wangIT10}\cite{HarveyIT}
and undirected networks \cite{LiLiunicast} in previous work. The
work of \cite{yan06isit}, provides an information theoretic
characterization for directed acyclic networks. However, in
practice, evaluating these bounds becomes computationally
infeasible even for small networks because of the large number of
inequalities that are involved. Moreover, these approaches do not
suggest any constructive code design approaches. The work of
\cite{wangIT10}, considers the multiple unicast problem in the
case of two source-terminal pairs, while the work of
\cite{Traskov06} attempts to address it by packing butterfly
networks within the original graph. Das et al. \cite{JafarISIT} consider the multiple unicast problem with an interference alignment approach. For undirected networks, there
is open conjecture as to whether there is any advantage to using
network coding as compared to routing (\cite{LiLiunicast}). Multiple unicast in the presence of link faults and errors, under certain restricted (though realistic) network topologies has been studied in \cite{kamal11}\cite{5205599}.

In this work our aim is to better understand the combinatorial
aspects of the multiple unicast problem over directed acyclic
networks. We consider a network $G$, with unit capacity edges and
source-terminal pairs, $s_i-t_i, i = 1,\dots,n$, such that the
maximum flow from $s_i$ to $t_i$ is $k_i$. Each source contains a
unit-entropy message that needs to be communicated to the
corresponding terminal. Our objective is to determine whether
there exist feasible network codes that can satisfy the demands of
the terminals. This is motivated by a need to find
characterizations that can be determined in a computationally
efficient manner.
%Of course, one could perhaps also characterize this region by considering the maximum flow between various $s_i - t_j$ pairs (for $i \neq j$). This is the subject of ongoing work. In this paper, we mostly focus on the case of $n = 3$.
\subsection{Main Contributions}
\begin{itemize}
\item For the case of three unicast sessions ($n = 3$), we
identify all feasible and infeasible connectivity levels $[k_1~k_2
~k_3]$. For the feasible cases, we provide efficient linear
network code assignments. For the infeasible cases, we provide
counter-examples, i.e., instances of graphs where the multiple
unicast cannot be supported under any (potentially nonlinear)
network coding scheme. \item We identify certain
feasible/infeasible instances with two unicast sessions, where the
message entropies are different. These are used to arrive at
conclusions for the problem in the case of higher $n ~(> 3)$.
\end{itemize}
This paper is organized as follows. In section \ref{sec:pre}, we introduce
several concepts that will be used throughout the paper. We also describe
the precise problem formulation. Section \ref{sec:vector_routing}
identifies the feasible routing connectivity levels. We discuss the network coding case in
Section \ref{sec:netcod}. Counter examples are given for infeasible connectivity levels.
A feasible connectivity level with vector network coding solution is also provided.
Section \ref{sec:con} concludes the paper.

%\section{Background and Related Work}
\section{Preliminaries}
\label{sec:pre}
We represent the network as a directed acyclic graph $G = (V, E)$.
Each edge $e \in E$ has unit capacity and can transmit one symbol
from a finite field of size $q$ per unit time (we are free to
choose $q$ large enough). If a given edge has higher capacity, it
can be treated as multiple unit capacity edges. A directed edge
$e$ between nodes $i$ and $j$ is represented as $(i,j)$, so that
$head(e) = j$ and $tail(e) = i$. A path between two nodes $i$ and
$j$ is a sequence of edges $\{ e_1, e_2, \dots, e_k\}$ such that
$tail(e_1) = i, head(e_k) = j$ and $head(e_i) = tail(e_{i+1}), i =
1, \dots, k-1$. The network contains a set of $n$ source nodes
$s_i$ and $n$ terminal nodes $t_i, i = 1, \dots
n$. Each source node $s_i$ observes a discrete integer-entropy
source, that needs to be communicated to
terminal $t_i$. Without loss of generality, we assume that the
source (terminal) nodes do not have incoming (outgoing) edges. If
this is not the case one can always introduce an artificial source
(terminal) node connected to the original source (terminal) node
by an edge of sufficiently large capacity that has no incoming
(outgoing) edges.

%Furthermore, we also consider the vector source $\vec{X}_i$, to be a set of independent collocated unit-entropy sources at node $s_i$.

We now discuss the network coding model under consideration in
this paper. For the sake of simplicity, suppose that each source
has unit-entropy, denoted by $X_i$. In scalar linear network coding, the
signal on an edge $(i,j)$, is a linear combination of the signals
on the incoming edges on $i$ or the source signals at $i$ (if $i$
is a source). We shall only be concerned with networks that are
directed acyclic and can therefore be treated as delay-free
networks \cite{rm}. Let $Y_{e_i}$ (such that $tail(e_i) = k$ and
$head(e_i) = l$) denote the signal on edge $e_i \in E$. Then, we
have
\begin{align*}
Y_{e_i} &= \sum_{\{e_j | head(e_j) = k\}} f_{j,i} Y_{e_j} \text{~if $k \in V \backslash \{s_1, \dots, s_n\}$}, \text{~and}\\
Y_{e_i} &= \sum_{j=1}^n a_{j,i} X_j
\text{~~ where $a_{j,i} = 0$ if $X_j$ is not observed at $k$.}
\end{align*}
The coefficients $a_{j,i}$ and $f_{j,i}$ are from the operational field.
Note that since the graph is directed acyclic, it is equivalently possible to
express $Y_{e_i}$ for an edge $e_i$ in terms of the sources
$X_j$'s. If $Y_{e_i} = \sum_{k=1}^n \beta_{e_i, k} X_k$ then we say that the
global coding vector of edge $e_i$ is $\boldsymbol{\beta}_{e_i} =
[ \beta_{e_i, 1} ~\cdots~ \beta_{e_i, n}]$. We shall also
occasionally use the term coding vector instead of global coding
vector in this paper. We say that a node $i$ (or edge $e_i$) is
downstream of another node $j$ (or edge $e_j$) if there exists a
path from $j$ (or $e_j$) to $i$ (or $e_i$).

Vector linear network coding is a generalization of the scalar
case, where we code across the source symbols in time, and the
intermediate nodes can implement more powerful operations.
Formally, suppose that the network is used over $T$ time units. We
treat this case as follows. Source node $s_i$ now observes a
vector source $[X_i^{(1)} ~ \dots ~ X_i^{(T)}]$. Each edge in the
original graph is replaced by $T$ parallel edges. In this graph,
suppose that a node $j$ has a set of $\beta_{inc}$ incoming edges
over which it receives a certain number of symbols, and $\beta_{out}$
outgoing edges. Under vector network coding, $j$ chooses
a matrix of dimension $\beta_{out} \times \beta_{inc}$. Each row
of this matrix corresponds to the local coding vector of an
outgoing edge from $j$.

Note that the general multiple unicast problem, where edges have
different capacities and the sources have different entropies can
be cast in the above framework by splitting higher capacity edges
into parallel unit capacity edges, a higher entropy source into
multiple, collocated unit-entropy sources; and the corresponding
terminal node into multiple, collocated terminal nodes.

An instance of the multiple unicast problem is specified by the
graph $G$ and the source terminal pairs $s_i - t_i,  i = 1, \dots,
n$, and is denoted $<G, \{s_i - t_i\}_{1}^n, \{R_i\}_{1}^n>$,
where the rates $R_i$ denote the entropy of the $i^{th}$ source.
For convenience, if all the sources are unit entropy, we will
refer to the instance by just $<G, \{s_i - t_i\}_{1}^n>$, where
the $s_i-t_i$ connections will occasionally be referred to as
sessions that we need to support.
%In this case we say that there are $n$ unicast sessions, indexed by $i = 1, \dots, n$ that we need to support.

The instance is said to have a scalar linear network coding
solution if there exist a set of linear encoding coefficients for
each node in $V$ such that each terminal $t_i$ can recover $X_i$
using the received symbols at its input edges. Likewise, it is
said to have a vector linear network coding solution with vector
length $T$ if the network employs vector linear network codes and
each terminal $t_i$ can recover $[X_i^{(1)} ~ \dots ~ X_i^{(T)}]$.

We will also be interested in examining the existence of a routing
solution, wherever possible. In a routing solution, each edge
carries a copy of one of the sources, i.e., each coding vector is
such that at most one entry takes the value $1$, all others are
$0$. Scalar (vector) routing solutions can be defined in a
manner similar to scalar (vector) network codes. We now define
some quantities that shall be used throughout the paper.

\begin{definition}
{\it Connectivity level.} The connectivity level for
source-terminal pair $s_i - t_i$ is said to be $n$ if the maximum
flow between $s_i$ and $t_i$ in $G$ is $n$. The connectivity level
of the set of connections $s_1 - t_1, \dots, s_n - t_n$ is the
vector $[ \text{max-flow}(s_1 - t_1) ~ \text{max-flow}(s_2 - t_2)~
\dots ~ \text{max-flow}(s_n - t_n)]$.
\end{definition}

In this work our aim is to characterize the feasibility of the multiple unicast problem based on the connectivity level of the $s_i -t_i$ pairs. The questions that we seek to answer are of the following form.\\
Suppose that the connectivity level is $[k_1~ k_2~ \dots~ k_n]$. Does any instance always have a linear (scalar or vector) network coding solution? If not, is it possible to demonstrate a counter-example, i.e,  an instance of a graph $G$ and $s_i - t_i$'s such that recovering $X_i$ at $t_i$ for all $i$ is impossible under linear (or nonlinear) strategies?

In this paper, our achievability results will be constructive and
based on linear network coding, whereas the counter-examples will
hold under all possible strategies.

\section{Multiple unicast under routing}
\label{sec:vector_routing}
We begin by providing a simple condition that guarantees the existence of a routing solution.
\begin{theorem}
Consider a multiple unicast instance with $n~$ $s_i - t_i$ pairs such that the connectivity level is $[n ~n ~ \dots ~n]$. There exists a vector routing solution for this instance.
\end{theorem}
\begin{proof}
Under vector routing over $n$ time units, source $s_i$ observes
$[X_i^{(1)} ~ \dots ~ X_i^{(n)}]$ symbols. Each edge $e$ in the
original graph is replaced by $n$ parallel edges, $e^1, e^2,
\dots, e^n$. Let $G_{\alpha}$ represent the subgraph of this graph
consisting of edges with superscript $\alpha$. It is evident that
max-flow($s_{\alpha} - t_{\alpha}$) = $n$ over $G_{\alpha}$. Thus,
we transmit $X_{\alpha}^{(1)}, \dots, X_{\alpha}^{(n)}$ over $G_{\alpha}$
using routing, for all $\alpha=1,\dots,n$. It is clear that this strategy satisfies the
demands of all the terminals.
%
% It is equivalent to view the original graph $G$ as $n$ graphs ($G_i$, $i=1,\cdots,n$) with the same structure. In $G_i$, $s_i$ transmits $[X_i^{(1)} ~ \dots ~ X_i^{(n)}]$. Because the connectivity level for $s_i - t_i$ is $n$, $t_i$ can recover $n$ symbols with routing. Hence, during $n$ time units, each terminal can recover $n$ symbols.
%\aditya{Shurui, include while keeping in mind the formal problem specification above.}
\end{proof}
Note that in general, a network with the above connectivity level may not be able to support a scalar routing solution, an instance is shown in Figure \ref{fig:2paths_frac_routing}. However, a scalar network coding solution exists for this example.

%\begin{figure}[t]
%\begin{center}
%\includegraphics[width=25mm,clip=false, viewport=80 0 200 280]{2path_revised.eps}
%\caption{A network with connectivity levels $[2~2]$ and rate $\{1,1\}$. There is a vector routing solution as shown in the figure.
%There is no scalar routing solution.} \label{fig:2paths_frac_routing}\centering
%\end{center}
%\end{figure}

\begin{figure}[t]
\begin{center}
\includegraphics[width=55mm,clip=false, viewport=0 110 260 380]{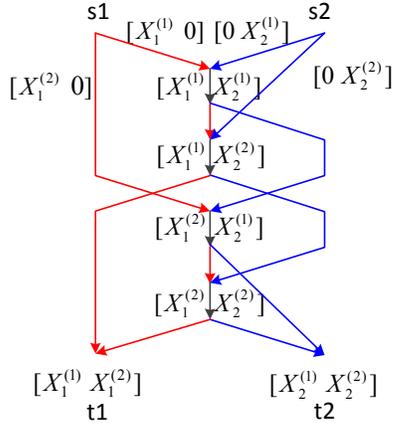}
\caption{A network with connectivity levels $[2~2]$ and rate $\{1,1\}$. There is a vector routing solution as shown in the figure.
There is no scalar routing solution.} \label{fig:2paths_frac_routing}\centering
\end{center}
\end{figure}

%\aditya{Fix all the figures}
%\shurui{I will modify the labels on the figure later, the structure of the network is like this.}
%\aditya{We do have a simple claim that there is an instance such that this cannot be done with scalar routing. Probably need to mention that.}

\section{Network coding for three unicast sessions}
\label{sec:netcod}
%\aditya{General comment: check all figures and captions and put
%arrows on edges.}
In the case of three unicast sessions, it is
clear based on the results of Section \ref{sec:vector_routing}
that if the connectivity level is $[3 ~3 ~3]$, then a vector
routing solution always exists. In this section we provide a full
characterization of the feasibility/infeasibility of supporting
three unicast sessions for a connectivity level of $[k_1 ~ k_2 ~
k_3]$, where $1 \leq k_i, \leq 3, i = 1, \dots, 3$. For the
feasible cases we will demonstrate appropriate linear network code
assignments. On the other hand, for the infeasible cases we will
present counter-examples where it is not possible to satisfy the
terminal's demands under any coding strategy.

\subsection{Infeasible Instances}
%The capacity region of a network is difficult to compute. However, an outer bound is computed by the cut set bound. Based on this observation, we have the following Lemma.
We begin by demonstrating certain instances that can be ruled out by using cutset bounds.
\begin{lemma}
\label{lemma:meagerness_cases}
There exist multiple unicast instances with three unicast sessions such that the connectivity levels $[2~2~2]$ and $[1~1~3]$ are infeasible.
\end{lemma}
\begin{proof}
A network with connectivity levels $[2~2~2]$ is shown in Figure \ref{fig:3sources2paths}. Consider the cut specified by the set of nodes $\{s_1,s_2,s_3,v_1,v_2\}$ that has a capacity value of 2. %\aditya{relabel the nodes appropriately} %The mutual information $I(X_S;Y_{S^C}|X_{S^C})$ is bounded by 2 where $X_S$ is the information at $s_1$, $s_2$ and $s_3$, and $Y_{S^C}$ is the information at $q_3$ and $q_4$.
The rate that needs to be supported over $\{e_1,e_2\}$ is 3. By the cut set bound, this rate cannot be achieved.
%By the meagerness bound, if edges $e_1$ and $e_2$ are removed, $I=\{1,2,3\}$ will be isolated, $M(e_1,e_2)=2/3$. The network coding capacity is bounded by $2/3$. Therefore, there does not exist a network coding solution.

\begin{figure}[htbp]
\hspace{-0.1in}\center{
\subfigure[]{\label{fig:3sources2paths}
\includegraphics[width=33mm,clip=false, viewport=50 50 150 180]{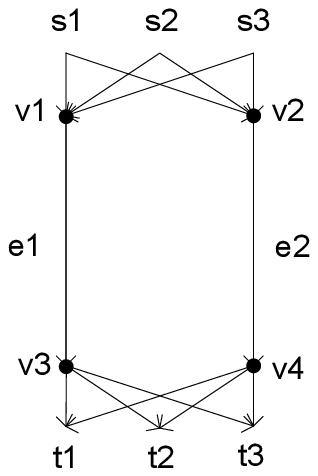}}
%\vspace{0.2in}
\hspace{0.4in}
\subfigure[]{\label{fig:counter311paths}
\includegraphics[width=40mm,clip=false, viewport=40 40 220 230]{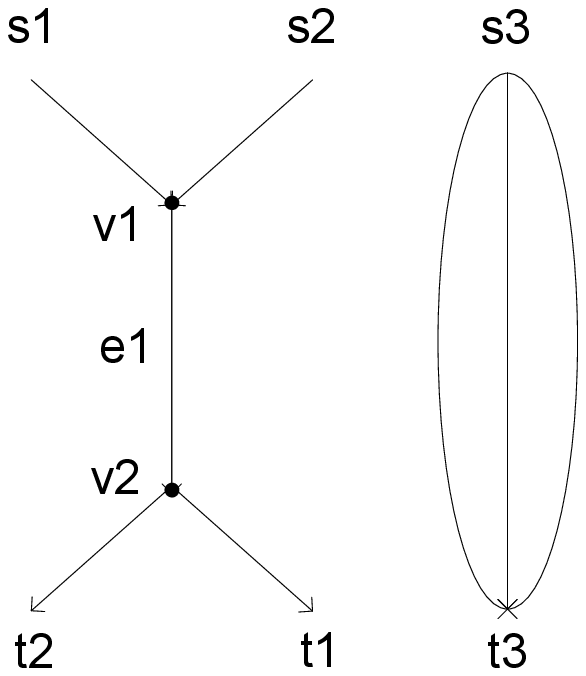}}
\caption{\label{fig: network_structure} (a) An example of $[2~2~2]$ connectivity network without a network coding solution.
(b) An example of $[1~1~3]$ connectivity network without a network coding solution.}}
\end{figure}

Similarly, a network with connectivity levels $[1~1~3]$ is shown in Figure \ref{fig:counter311paths}.
Consider the cut $\{s_1,s_2,v_1\}$. The capacity of this cut is 1. %The mutual information $I(X_S;Y_{S^C}|X_{S^C})$ is bounded by 1 where $X_S$ is the information at $s_1$, $s_2$, and $Y_{S^C}$ is the information at $q_2$.
However, the rate that needs to be supported over $e_1$ is 2. Therefore, there does not exist a network coding solution.
%By the meagerness bound, if edges $e_1$ is removed, $I=\{1,2\}$ will be isolated, $M(e_1)=1/2$. The network coding capacity is bounded by $1/2$. Therefore, there does not exist a network coding solution.
\end{proof}
While cut set bound is useful in the above cases, there exist certain connectivity levels for which a cut set bound is not tight enough.
We now present such an instance in Figure \ref{fig:21counter}. We show that this instance is not feasible under any code scheme (linear or nonlinear). This instance was also presented in the work of Erez and Feder \cite{feder09}, though they did not provide a formal proof of this fact. %) where the connectivity level is $[2 ~ 2 ~3]$, yet the instance is infeasible\footnote{It is clear that the same instance can be a counter example for the cases considered in Lemma \ref{lemma:meagerness_cases} above. However, we presented those cases separately, since the arguments follow in a straightforward manner from meagerness bounds. However, the cut set bound does not apply in a straightforward manner for the counter-example in Figure \ref{fig:21counter}.}. %It can be verified, that applying the meagerness bound does not rule out the feasibility of this instance. However, we have the following result.
\begin{lemma}
\label{lemma:21counter}
There exists a multiple unicast instance, with two sessions $<G,\{s_1-t_1, s_2-t_2\}, \{2, 1\}>$ and connectivity level $[2~3]$ that is infeasible.
\end{lemma}
\begin{proof}
The graph instance is shown in Figure \ref{fig:21counter}. Assume in $n$ time units, $s_1$ observes two independent vector sources $[X_1^{(1)} ~ \dots ~ X_1^{(n)}]$ and $[X_2^{(1)} ~ \dots ~ X_2^{(n)}]$, $s_2$ observes one independent vector source $[X_3^{(1)} ~ \dots ~ X_3^{(n)}]$. The sources are denoted as $X_1^n$, $X_2^n$ and $X_3^n$ for simplicity. The $n$ random variables that $e_i$ carries are denoted as $Y_{e_i}^n$, or simply $Y_i^n$.
Suppose that the alphabet of $X_i$ is $\mathcal{X}$. Since the entropy rates for the three sources are the same,
we can assume $H(X_i)= \log|\mathcal{X}|=a$. Also, since we are interested in
the feasibility of the solution, we can further assume that the alphabet size of $Y_{ij}$
is also the same as $\mathcal{X}$, and $H(Y_{ij})\leq\log|\mathcal{X}|=a$ by the capacity constraint of the edge.
At terminal $t_1$ and $t_2$, from $Y^n_{11}$, $Y^n_{12}$, $Y^n_{21}$ and $Y^n_{22}$,
we estimate $X^n_1$, $X^n_2$ and $X^n_3$.
Let the estimate be $\widehat{X}^n_1$, $\widehat{X}^n_2$ and $\widehat{X}^n_3$.
Suppose that there exist network codes and decoding function such that
$P((\widehat{X}^n_1,\widehat{X}^n_2)\neq (X_1^n,X_2^n))\rightarrow 0$ as $n\rightarrow \infty$.
\begin{figure}[t]
\begin{center}
\includegraphics[width=40mm,clip=false, viewport=0 0 150 215]{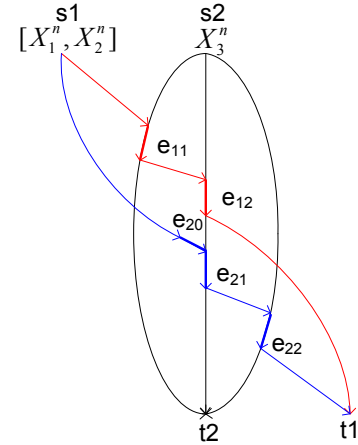}
\caption{An example of $[2~3]$ connectivity network, rate $\{2,1\}$
cannot be supported.} \label{fig:21counter}\centering
\end{center}
\end{figure}
From the Fano's inequality, we shall have
\begin{equation}
\label{eq:fano}
 H(X_1^n,X_2^n|\widehat{X}^n_1,\widehat{X}^n_2)\leq n\ep.
\end{equation}
where $n\ep=1+nP_e\log(|\mathcal X|)$. For $t_1$ to decode $X_1^n$
and $X_2^n$ asymptotically, $\ep\rightarrow0$ as $P_e\rightarrow0$,
when $n\rightarrow\infty$, where $P_e=P((\widehat{X}^n_1,\widehat{X}^n_2)\neq (X_1^n,X_2^n))$.

Likewise, decodability at $t_1$ implies that $\widehat{X}^n_1,\widehat{X}^n_2$ are functions of $Y^n_{12}$ and $Y^n_{22}$. Hence, we will have
\begin{equation}
\label{eq:fano2}
\begin{split}
H(X_1^n,X_2^n|Y^n_{12},Y^n_{22})&=H(X_1^n,X_2^n|\widehat{X}^n_1,\widehat{X}^n_2,Y^n_{12},Y^n_{22})\\
&\leq H(X_1^n,X_2^n|\widehat{X}^n_1,\widehat{X}^n_2)\leq n\ep.
\end{split}
\end{equation}
Now the sequences of information coming into $t_1$ are,
\begin{equation}
\label{eq:begin}
\begin{split}
2an&\stackrel{_{(a)}}{\geq} H(Y^n_{12},Y^n_{22})\\
&\stackrel{_{(b)}}{=}H(Y^n_{12},Y^n_{22},X_1^n,X_2^n)-H(X_1^n,X_2^n|Y^n_{12},Y^n_{22})\\
&\geq H(X_1^n,X_2^n)-H(X_1^n,X_2^n|Y^n_{12},Y^n_{22})\\
&\stackrel{_{(c)}}{\geq}2an-n\ep
\end{split}
\end{equation}
(a) is due to the capacity constraints of the edge $e_{12}$ and
$e_{22}$. (b) follows from the chain rule. (c) is because rate
$2an$ should be transmitted over $n$ time units and Equation (\ref{eq:fano2}).
%Also from the
%capacity constraints, we shall have
%\begin{equation}
%n\geq H(Y^n_{12})\geq n-n\ep,~~~n\geq H(Y^n_{22})\geq
%n-n\ep;~~~2n\geq H(Y^n_{12},Y^n_{22})\geq 2n-n\ep
%\end{equation}

Next, we shall have
\begin{equation}
\label{eq:funcofab}
\begin{split}
&H(Y^n_{12},Y^n_{22}|X_1^n,X_2^n)\\
&\stackrel{_{(a)}}{=}H(Y^n_{12},Y^n_{22},X_1^n,X_2^n)-H(X_1^n,X_2^n)\\
&\stackrel{_{(b)}}{=}H(X_1^n,X_2^n|Y^n_{12},Y^n_{22})+H(Y^n_{12},Y^n_{22})-H(X_1^n,X_2^n)\\
&\stackrel{_{(c)}}{\leq} n\ep+2an-2an=n\ep.
\end{split}
\end{equation}
(a)(b) follows from the chain rule. (c) is from Equation
(\ref{eq:fano2}) and Equation (\ref{eq:begin}).

%Because $\ep\rightarrow 0$, $Y^n_{12}$ and $Y^n_{22}$ are
%functions of $X_1^n$ and $X_2^n$.

Analyzing the independence of $X_1^n$, $X_2^n$, and $X_3^n$, we shall
have
\begin{equation}
\begin{split}
an&=H(X_3^n|X_1^n,X_2^n)\\
&\stackrel{_{(a)}}{=}H(X_3^n|Y^n_{12},Y^n_{22},X_1^n,X_2^n)+I(X_3^n;Y^n_{12},Y^n_{22}|X_1^n,X_2^n)\\
&=H(X_3^n|Y^n_{12},Y^n_{22},X_1^n,X_2^n)+H(Y^n_{12},Y^n_{22}|X_1^n,X_2^n)\\
&~~~~~~-H(Y^n_{12},Y^n_{22}|X_1^n,X_2^n,X_3^n)\\
&\stackrel{_{(b)}}{\leq} H(X_3^n|Y^n_{12},Y^n_{22},X_1^n,X_2^n)+n\ep\\
&\stackrel{_{(c)}}{\leq}
H(X_3^n|Y^n_{12},Y^n_{22})+n\ep\stackrel{_{(d)}}{\leq} an+n\ep
\end{split}
\end{equation}
(a) is from the definition of conditional mutual information. (b)
is from Equation (\ref{eq:funcofab}) and because conditioning
reduces entropy. (c) is because conditioning reduces entropy. (d)
is because conditioning reduces entropy. From the above
inequalities, the information on $e_{12}$ and $e_{22}$ cannot
decode $X_3^n$ asymptotically. Then we have the following equations,
\begin{equation}
\label{eq:cn}
an-n\ep\leq H(X_3^n|Y^n_{12},Y^n_{22})\leq an
\end{equation}
\begin{equation}
I(Y^n_{12},Y^n_{22};X_3^n)=H(X_3^n)-H(X_3^n|Y^n_{12},Y^n_{22})\leq n\ep;
\end{equation}
\begin{equation}
\label{eq:diff}
\begin{split}
H(Y^n_{12},Y^n_{22}|X_3^n)&=H(Y^n_{12},Y^n_{22})-I(Y^n_{12},Y^n_{22};X_3^n)\\
&\geq 2an-2n\ep
\end{split}
\end{equation}
\begin{equation*}
I(Y^n_{12};X_3^n)=I(Y^n_{12},Y^n_{22};X_3^n)-I(Y^n_{22};X_3^n|Y^n_{12})\leq
n\ep;
\end{equation*}
\begin{equation}
I(Y^n_{22};X_3^n)\leq n\ep
\end{equation}

The above inequalities imply that the information on $e_{12}$ and
$e_{22}$ are asymptotically independent of $X_3^n$.
%
%Because $\ep\rightarrow 0$, $Y^n_{12}$ and $Y^n_{22}$ are
%functions of $X_1^n$ and $X_2^n$. $Y^n_{12}$ and $Y^n_{22}$ are only
%functions of $X_1^n$ and $X_2^n$, and are independent of $X_3^n$.

Because $Y^n_{21}$ is only a function of $Y^n_{12}$ and
$Y^n_{20}$, %$Y^n_{21}$ is only a function of $X_1^n$ and $X_2^n$.
\begin{equation}
\label{eq:same}
\begin{split}
&H(Y^n_{21},Y^n_{22})\\
&\stackrel{_{(a)}}{=}H(X_3^n,Y^n_{21},Y^n_{22})-H(X_3^n|Y^n_{21},Y^n_{22})\\
&\stackrel{_{(b)}}{=}H(X_3^n,Y^n_{21})-H(X_3^n|Y^n_{21},Y^n_{22})\\
&\stackrel{_{(c)}}{\leq}2an-H(X_3^n|Y^n_{21},Y^n_{22})\\
&\stackrel{_{(d)}}{\leq}2an-H(X_3^n|Y^n_{21},Y^n_{22},Y^n_{20},Y^n_{12},X_1^n,X_2^n)\\
&\stackrel{_{(e)}}{=}2an-H(X_3^n|Y^n_{22},Y^n_{20},Y^n_{12},X_1^n,X_2^n)\\
&\stackrel{_{(f)}}{=}2an-H(X_3^n|Y^n_{22},X_1^n,X_2^n,Y^n_{12})\\
&\stackrel{_{(g)}}{=}2an-H(X_3^n|Y^n_{22},Y^n_{12})+I(X_3^n;X_1^n,X_2^n|Y^n_{22},Y^n_{12})\\
&\stackrel{_{(h)}}{=}2an-H(X_3^n|Y^n_{22},Y^n_{12})+H(X_1^n,X_2^n|Y^n_{22},Y^n_{12})\\
&~~~~~~-H(X_1^n,X_2^n|Y^n_{22},X_3^n,Y^n_{12})\\
&\leq 2an-H(X_3^n|Y^n_{22},Y^n_{12})+H(X_1^n,X_2^n|Y^n_{22},Y^n_{12})\\
&\stackrel{_{(i)}}{\leq}2an-an+n\ep+n\ep=an+2n\ep
\end{split}
\end{equation}
(a) follows from the chain rule, (b) is because $Y^n_{22}$ is a
function of $X_3^n$ and $Y^n_{21}$. (c) is because of the capacity
constraints. (d) is because conditioning reduces entropy. (e) is
because  $Y^n_{21}$ is a function of $Y^n_{12}$ and $Y^n_{20}$.
(f) is because $Y^n_{20}$ is a function of $X_1^n$ and $X_2^n$. (g)(h)
follows from the mutual information definition. (i) is from
Equation (\ref{eq:fano2}) and Equation (\ref{eq:cn}). The above
inequalities indicate that $e_{21}$ and $e_{22}$ should carry the
same information asymptotically.

From the network, we know that $Y^n_{12}$ is a function of
$Y^n_{11}$ and $X_3^n$. Then
\begin{equation}
\label{eq:indep}
\begin{split}
H(Y^n_{11},&Y^n_{21},Y^n_{22}|X_3^n)=H(Y^n_{11},Y^n_{21},Y^n_{22},X_3^n|X_3^n)\\
&\geq H(Y^n_{12},Y^n_{21},Y^n_{22}|X_3^n)\\
&\geq H(Y^n_{22},Y^n_{12}|X_3^n)\stackrel{_{(a)}}{\geq} 2an-2n\ep
\end{split}
\end{equation}
(a) is due to Equation (\ref{eq:diff}).

Finally, we shall have
\begin{equation}
\begin{split}
&H(X_3^n|Y^n_{11},Y^n_{21},Y^n_{22})\\
&=H(Y^n_{11},Y^n_{21},Y^n_{22}|X_3^n)+H(X_3^n)-H(Y^n_{22},Y^n_{21},Y^n_{11})\\
&\stackrel{_{(a)}}{\geq}2an-2n\ep+an-H(Y^n_{22},Y^n_{21},Y^n_{11})\\
&=3an-2n\ep-H(Y^n_{22},Y^n_{21})-H(Y^n_{11}|Y^n_{22},Y^n_{21})\\
&\stackrel{_{(b)}}{\geq}3an-2n\ep-an-2n\ep-H(Y^n_{11}|Y^n_{22},Y^n_{21})\\
&\stackrel{_{(c)}}{\geq}2an-4n\ep-an=an-4n\ep
\end{split}
\end{equation}
(a) is because of Equation (\ref{eq:indep}). (b) is because of
Equation (\ref{eq:same}). (c) is due to the capacity constraint of
$Y^n_{11}$.

When $n\rightarrow \infty$, for $t_1$ to asymptotically decode
$X_1^n$ and $X_2^n$, we shall have $\ep\rightarrow 0$. Then $t_2$
cannot decode $X_3^n$ asymptotically.
\end{proof}
\begin{corollary}
There exists a multiple unicast instance with three sessions, and connectivity level $[2~3~2]$ that is infeasible.
\end{corollary}
\begin{proof}
%For network $<G,\{s_i-t_i\}^3_1,\{1,1,1\}>$, if $s_1$ and $s_3$
%are collocated at $s_1$, $t_1$ and $t_3$ are collocated at $t_1$,
%and $<G,\{s_1-t_1,s_2-t_2\},\{2,1\}>$ forms the structure in Lemma
%\ref{lemma:21counter}, the three sessions will have connectivity levels $[2~3~2]$, and there is no feasible solution for
%$<G,\{s_i-t_i\}^3_1,\{1,1,1\}>$.
Consider a multiple unicast instance $<G,\{s'_i-t'_i\}^3_1,\{1,1,1\}>$, where $G$ is the graph in Figure \ref{fig:21counter}.
The sources $s'_1$ and $s_3'$ are collocated at $s_1$ (in $G$), and the terminals $t_1'$ and $t_3'$ are collocated
at $t_1$ (in $G$). Likewise, the source $s_2'$ and terminal $t_2'$ are located at $s_2$ and $t_2$ in $G$.
The three sessions have connectivity level $[2~3~2]$. Based on the arguments in Lemma \ref{lemma:21counter},
there is no feasible solution for this instance.
\end{proof}

The instance presented in Lemma \ref{lemma:21counter}, can be generalized to obtain a series of counter-examples.
In particular, we have the following theorem shows an instance with two unicast sessions with connectivity level $[n_1~n_2]$ that cannot support rates $R_1 = n_1, R_2 = n_2 - n_1$.
\begin{theorem}
\label{lemma:counterforn} For a directed acyclic graph $G$ with
two $s-t$ pairs, if the connectivity level for $(s_1,t_1)$ is
$n_1$, for $(s_2,t_2)$ is $n_2$, $1<n_1<n_2$, there exist instances
that cannot support $R_1=n_1$ and $R_2=n_2-n_1$.
\end{theorem}
\begin{proof}
The proof is omitted due to space limitations.
\end{proof}

\subsection{Feasible Instances}
It is evident that the infeasibility of an instance with connectivity level $[2~2~3]$ implies that when $1\leq k_i\leq 3$, the only possible instances that are potentially feasible are $[1~3~3]$, its permutations and connectivity levels that are greater than it. We now show that many of these instances are feasible using linear network codes. In this subsection, we present efficient linear network code assignment algorithms for these cases. Towards this end, we need the following definitions.

\begin{definition} {\it Minimality.} Consider a multiple unicast instance $<G = (V,E), \{s_i - t_i\}_{1}^{n}>$, with connectivity level $[k_1~k_2~\dots~ k_n]$. The graph $G$ is said to be minimal if the removal of any edge from $E$ strictly reduces the connectivity level. If $G$ is minimal, we will also refer to the multiple unicast instance as minimal.
\end{definition}
Clearly, given a non-minimal instance $G = (V,E)$, we can always remove the non-essential edges from it, to obtain the minimal graph $G_{\min}$. This does not affect feasibility, since a network code for $G_{\min} = (V, E_{\min})$ can be converted into a network code for $G$ by simply assigning the all-zeros coding vector to the edges in $E \backslash E_{\min}$.
\begin{definition}
{\it Overlap edge.} An edge $e$ is said to be an overlap edge for paths $P_i$ and $P_j$ in $G$, if $e \in P_i \cap P_j$. %\shurui{$P_i$ and $P_j$ is not defined early, I think it is important to point out they are for $s_i-t_i$ and $s_j-t_j$. I agree here we don't need minimal graph condition. }\aditya{Not clear to me why you specify that $G$ is minimal in your draft.}
\end{definition}
\begin{definition} {\it Overlap segment.} In $G$, consider an ordered set of edges $E_{os} = \{e_1, \dots, e_l\} \subset E$ that forms a path. This path is called an overlap segment for paths $P_i$ and $P_j$ if
\begin{itemize}
%\item[(i)] Each edge $e_i, i = 1, \dots, l$ is an overlap edge for $P_i$ and $P_j$.
\item[(i)] $\forall k \in\{1, \dots, l\}$, $e_k$ is an overlap edge for $P_i$ and $P_j$.
\item[(ii)] None of the incoming edges into tail($e_1$) are overlap edges for $P_i$ and $P_j$.
\item[(iii)] None of the outgoing edges leaving head($e_l$) are overlap edges for $P_i$ and $P_j$.
\end{itemize}
\end{definition}
Our solution strategy is as follows. We first convert the original instance into another {\it structured} instance where each internal node has at most degree three (in-degree + out-degree). We then convert this new instance into a minimal one, and then develop the code assignment algorithm. It will be evident that using this network code, one can obtain a network code for the original instance.

\subsubsection{Conversion procedure}
Let $G=(V,E)$ be our original graph, and
let $s_i$ and $t_i$ be the given sources and terminals.
We can efficiently construct a {\em structured} graph $\hat{G}=(\hat{V},\hat{E})$ in which each internal node $v \in \hV$ is of total degree at most three with the additional following properties:
(a) $\hG$ is acyclic.
(b) For every source (terminal) in $G$ there is a corresponding source (terminal) in $\hG$.
(c) For any two edge disjoint paths $P_i$ and $P_j$ for one unicast session in $G$, there exist two {\em vertex} disjoint paths in $\hG$ for the corresponding session in $\hG$.
(d) Any feasible network coding solution in $\hG$ can be efficiently turned into a feasible network coding solution in $G$.
Our reduction steps are the same as in \cite{LSB06}. Due to space limitations, refer to \cite{LSB06} and \cite{ramamoorthy-langberg} for details.

\subsubsection{Code Assignment Procedure}
In the discussion below, we will assume that the graph $G$ is structured. It is clear that this is without loss of generality based on the previous arguments. In our arguments, we will use the minimality of the graph extensively.
\begin{lemma}
Consider a minimal multiple unicast instance, $<G,\{s_1 - t_1, s_2 - t_2\}>$ with connectivity level $[1~m]$. Denote the $s_1 - t_1$ path by $P_1$ and the set of edge disjoint $s_2 - t_2$ paths by $\{P_{21}, \dots, P_{2m}\}$. There can be at most one overlap segment between $P_1$ and each $P_{2i}, i=1, \dots, m$.
\end{lemma}
\begin{proof}
Suppose that there are two overlap segments $E_{os1} = \{e_{1}, \dots, e_{k_1}\}$ and $E_{os2} = \{e'_{1}, \dots, e'_{k_2}\}$ between $P_1$ and $P_{2i}$, where $e_{k_1}$ is upstream of $e'_{1}$. Note that by the definition of an overlap segment and the fact that $G$ is structured, it holds that the head of $e_{k_1}$ has in-degree one and out-degree two, so that one outgoing edge from $head(e_{k_1})$ belongs to $P_1$ (denoted $e^*$) and the other belongs to $P_{2i}$. Note also $e^* \in P_1$ cannot belong to $P_{2j}, j \neq i$ since the set of paths $\{P_{21}, \dots, P_{2m}\}$ is vertex disjoint (since $G$ is structured).

Now, note that $e^*$ can be removed while still maintaining the required connectivity level. This is true for $s_2-t_2$, since $e^*$ does not lie on any of the paths $P_{21}, \dots, P_{2m}$. It is true for $s_1 - t_1$ since there is a path from $e_{k_1}$ to $e'_{k_2}$ that overlaps $P_{2i}$, and therefore this still continues be a path from $s_1 - t_1$. This path can be explicitly specified as $path(s_1, head(e_{k_1})),path(e_{k_1}, e'_{k_2}), path(head(e'_{k_2}), t_1)$.
\end{proof}
Using this property, we can obtain the following result that holds for the case of two unicast sessions with the rate $\{1,m\}$.
\begin{lemma}
\label{lemma:uneven_two_unicast}
A minimal multiple unicast instance $<G,\{s_1 - t_1, s_2 - t_2\}, \{1, m\}>$ with connectivity level $[1~m+1]$ is always feasible.
\end{lemma}
\begin{proof}
We show that this can be achieved by using scalar linear network codes.  Let $P_1$ denote the path from $s_1 - t_1$ and $m+1$ vertex-disjoint paths from $s_2 - t_2$, as $P_{2j}, j = 1, \dots, m+1$. Let the source message at $s_1$ be denoted by $X_1$ and the source message vector at $s_2$ by $[X_{21}, \dots, X_{2m}]$. We proceed by induction on $m$.

\noindent {\it Base case - $m=1$.} In this case suppose that $P_1$ intersects at most one path from the $s_2 - t_2$. For instance if $P_1$ overlaps with $P_{21}$, then simply transmit $X_{21}$ over $P_{22}$ and $X_1$ over $P_1$.

Alternatively, $P_1$ overlaps both $P_{21}$ and $P_{22}$. Suppose that the segments are denoted $E_{os1}$ and $E_{os2}$ respectively and that $E_{os1}$ is upstream of $E_{os2}$ (w.l.o.g.). In this case, we flow $X_1$ ($X_{21}$) on $P_1$ ($P_{21}$) until $E_{os1}$ and flow $X_1 + X_{21}$ on $E_{os1}$, and further downstream on $P_{21}$ till $t_2$ and on $P_1$ until $E_{os2}$. We flow $X_{21}$ on $P_{22}$ until $E_{os2}$ and flow $X_1 + X_{21} + X_{21} = X_1$, on $E_{os2}$ and further downstream till $t_1$ and $t_2$. It is evident that $t_2$ can recover $X_{21}$ from its received values.

\noindent {\it Induction step.} Suppose that the induction hypothesis holds for $m=k$. For $m = k+1$, again we consider two cases.
Suppose that $P_1$ does not overlap with at least one path from the set $\{P_{21}, \dots, P_{2 k+1}\}$,
w.l.o.g. suppose that it is $P_{2 k+1}$. In this case the graph consisting of $P_1 \cup P_{21} \cup \dots \cup P_{2k}$, can be used to transmit $X_1$ to $t_1$ and $X_{21},\dots , X_{2k-1}$ to $t_2$ using the induction hypothesis. $X_{2k}$ can simply be routed on $P_{2 k+1}$.

On the other hand if $P_1$ overlaps with all the paths $P_{21}, \dots, P_{2 k+1}$. We assume w.l.o.g. that it overlaps first with $P_{21}$ (in $E_{os1}$), then with $P_{22}$ and so on until $P_{2 k+1}$. In this case, as illustrated in Figure \ref{fig:1segment}, we can arrive at the required solution. In particular, $s_2$ transmits $X_{2i}$ over paths $P_{2i}, i = 1, \dots, k$ and $\sum_{j=1}^k X_{2j}$ over $P_{2k+1}$ until the overlap point. The path $P_1$ carries $X_1$ until $E_{os1}$. At each overlap segment a sum of the incoming values into the segment is computed. This ensures that overlap segment $E_{osi}$ carries $X_1 + \sum_{j=1}^i X_{2j}, i = 1, \dots, k$ and $E_{osk+1}$ carries $X_1$. It can be seen that both $t_1$ and $t_2$ are satisfied in this case.
\begin{figure}[t]
\begin{center}
\includegraphics[width=40mm,clip=false, viewport=0 0 150 210]{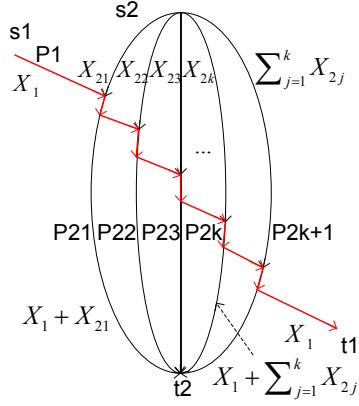}\vspace{-2mm}
\caption{An example where $P_1$ overlaps with all paths $P_{21}, \dots, P_{2 k+1}$. Rate $\{1,k\}$ is feasible.} \label{fig:1segment}\centering
\end{center}
\end{figure}
\end{proof}

It turns out that one can treat the case of three multiple unicast sessions with connectivity level $[1~3~3]$, by using the result of Lemma \ref{lemma:uneven_two_unicast}. The basic idea is to use vector linear network coding over two time units and code over pairs of sources at appropriately defined layers of the network. We state and prove this result below.

\begin{theorem}
\label{lemma:feasible}
A multiple unicast instance with three sessions such that the connectivity level is $[1 ~ 3 ~3]$ is always feasible.
\end{theorem}
\begin{proof}
Let the original graph (with unit capacity edges) be denoted by $G = (V,E)$. We use vector linear network coding over two time units, i.e. $T = 2$. In this case we form a new graph $G^*$ where each edge $e \in E$ is replaced by two parallel unit capacity edges $e^1$ and $e^2$ in $G^*$. The messages at source node $s_i$ are denoted $[X_{i1}~X_{i2}]$. Now consider the subgraph of $G^*$ induced by all edges with superscript $1$, that we denote $G^*_1$. In $G^*_1$, there exists a single $s_1 - t_1$ path and three edge disjoint $s_2 - t_2$ paths. Therefore, we can transmit $X_{11}$ from $s_1$ to $t_1$ and $[X_{21}~X_{22}]$ from $s_2 - t_2$ using the result of Lemma \ref{lemma:uneven_two_unicast}. Similarly, we use the subgraph of $G^*$ induced by all edges with superscript $2$, i.e., $G_2^*$ to communicate $X_{12}$ from $s_1$ to $t_1$ and $[X_{31}~X_{32}]$ from $s_3$ to $t_3$. Thus, using vector linear network coding over two time units, a connectivity level of $[1~3~3]$ suffices to satisfy the demands of each terminal.
\end{proof}

\begin{corollary}
A multiple unicast instance with three sessions such that the connectivity level is greater than $[1~3~3]$
is always feasible.
\end{corollary}
\begin{proof}
For the graph $G$ which has connectivity level greater than $[1~3~3]$, we identify a subgraph $G'$ with
connectivity level $[1~3~3]$. By Theorem \ref{lemma:feasible}, the demand at each terminal can be satisfied.
Then by assigning zero coding vector to the edges in $G\setminus G'$, the terminal demand can be satisfied
in the original graph $G$.
\end{proof}

So far, we have completely characterized the cases where the connectivity levels are $[k_1~k_2~k_3]$, $k_i\leq 3$.
However, there are several connectivity levels with unknown feasibility when $k_i>3$, e.g., $[2~2~4]$.
\section{Conclusions and Future Work}
\label{sec:con}

In this work, we have identified several feasible and infeasible
connectivity levels for 3 unicast sessions.
For the feasible instances, we provided explicit network code assignments,
while for the infeasible instances we demonstrated appropriate counter-examples.
Some of these results can be extended to the case of general $n$, and are currently under investigation.

\bibliographystyle{IEEEtran}
\bibliography{unicast}

\begin{thebibliography}{10}
\providecommand{\url}[1]{#1}
\csname url@samestyle\endcsname
\providecommand{\newblock}{\relax}
\providecommand{\bibinfo}[2]{#2}
\providecommand{\BIBentrySTDinterwordspacing}{\spaceskip=0pt\relax}
\providecommand{\BIBentryALTinterwordstretchfactor}{4}
\providecommand{\BIBentryALTinterwordspacing}{\spaceskip=\fontdimen2\font plus
\BIBentryALTinterwordstretchfactor\fontdimen3\font minus
  \fontdimen4\font\relax}
\providecommand{\BIBforeignlanguage}[2]{{%
\expandafter\ifx\csname l@#1\endcsname\relax
\typeout{** WARNING: IEEEtran.bst: No hyphenation pattern has been}%
\typeout{** loaded for the language `#1'. Using the pattern for}%
\typeout{** the default language instead.}%
\else
\language=\csname l@#1\endcsname
\fi
#2}}
\providecommand{\BIBdecl}{\relax}
\BIBdecl

\bibitem{rm}
R.~Koetter and M.~M\'{e}dard, ``{An Algebraic approach to network coding},''
  \emph{IEEE/ACM Trans. on Netw.}, vol. 11, no. 5, pp. 782--795, 2003.

\bibitem{al}
R.~Ahlswede, N.~Cai, S.-Y. Li, and R.~W. Yeung, ``{Network Information Flow},''
  \emph{IEEE Trans. on Info. Th.}, vol. 46, no. 4, pp. 1204--1216, 2000.

\bibitem{Traskov06}
D.~Traskov, N.~Ratnakar, D.~Lun, R.~Koetter, and M.~Medard, ``{Network Coding
  for Multiple Unicasts: An Approach based on Linear Optimization},'' in
  \emph{IEEE Intl. Symposium on Info. Th.}, 2006, pp. 1758--1762.

\bibitem{wangIT10}
C.-C. Wang and N.~B. Shroff, ``{Pairwise Intersession Network Coding on
  Directed Networks},'' \emph{IEEE Trans. on Info. Th.}, vol. 56, no. 8, pp.
  3879--3900, Aug, 2010.

\bibitem{HarveyIT}
N.~Harvey, R.~Kleinberg, and A.~Lehman, ``{On the capacity of information
  networks},'' \emph{IEEE Trans. on Info. Th.}, vol. 52, no. 6, pp. 2345--2364,
  2006.

\bibitem{LiLiunicast}
Z.~Li and B.~Li, ``{Network coding: the Case of Multiple Unicast Sessions},''
  in \emph{42st Allerton Conference on Communication, Control, and Computing},
  2004.

\bibitem{yan06isit}
X.~Yan, R.~W. Yeung, and Z.~Zhang, ``{The Capacity Region for Multi-source
  Multi-sink Network Coding},'' in \emph{IEEE Intl. Symposium on Info. Th.},
  June, 2007, pp. 116--120.

\bibitem{JafarISIT}
A.~Das, S.~Vishwanath, S.~A. Jafar, and A.~Markopoulou, ``{Network Coding for
  Multiple Unicasts: An Interference Alignment Approach},'' in \emph{IEEE Intl.
  Symposium on Info. Th.}, 2010, pp. 1878 -- 1882.

\bibitem{kamal11}
A.~E. Kamal, A.~Ramamoorthy, L.~Long, and S.~Li, ``{Overlay protection against
  link failures using network coding},'' \emph{IEEE/ACM Trans. on Netw.}, (to
  appear),~2011.

\bibitem{5205599}
S.~Li and A.~Ramamoorthy, ``Protection against link errors and failures using
  network coding in overlay networks,'' in \emph{Information Theory, 2009. ISIT
  2009. IEEE International Symposium on}, July 2009, pp. 986 --990.

\bibitem{feder09}
E.~Erez and M.~Feder, ``{Improving the Multicommodity Flow Rate with Network
  Codes for Two Sources},'' \emph{{IEEE JOURNAL ON SELECTED AREAS IN
  COMMUNICATIONS}}, vol. 27, no. 5, pp. 814--824, 2009.

\bibitem{LSB06}
M.~Langberg, A.~Sprintson, and J.~Bruck, ``{The encoding complexity of network
  coding},'' \emph{IEEE Trans. on Info. Th.}, vol. 52, no. 6, pp. 2368--2397,
  2006.

\bibitem{ramamoorthy-langberg}
A.~Ramamoorthy and M.~Langberg, ``{Communicating the sum of sources over a
  network},'' \emph{Submitted. (available at
  http://www.arxiv.org/abs/1001.5319)}.

\end{thebibliography}
\end{document}